\documentclass[reqno]{amsart}
\usepackage[utf8]{inputenc}
\usepackage{pmboxdraw}
\usepackage{amsmath}
\usepackage{amssymb}
\usepackage{xcolor}
\usepackage{graphicx}
\usepackage{float}
\usepackage[hidelinks]{hyperref}

\calclayout

\usepackage{cite}
\newtheorem{theorem}{Theorem}

\newcommand{\C}{\mathbb{C}}
\newcommand{\N}{\mathbb{N}}

\newcommand{\ler}[1]{\left( #1 \right)}
\newcommand{\lesq}[1]{\left[ #1 \right]}
\newcommand{\lers}[1]{\left\{ #1 \right\}}
\newcommand{\hohc}{\cH \otimes \cH^*}

\newcommand{\norm}[1]{\left|\left|#1\right|\right|}
\newcommand{\bra}[1]{\langle #1 |}
\newcommand{\ket}[1]{| #1 \rangle}
\newcommand{\Bra}[1]{\langle\langle #1 ||}
\newcommand{\Ket}[1]{|| #1 \rangle\rangle}

\newcommand{\be}{\begin{equation}}
\newcommand{\ee}{\end{equation}}
\newcommand{\ba}{\begin{array}}
\newcommand{\ea}{\end{array}}

\newcommand{\cH}{\mathcal{H}}
\newcommand{\cB}{\mathcal{B}}
\newcommand{\cK}{\mathcal{K}}

\newcommand{\cS}{\mathcal{S}}
\newcommand{\cC}{\mathcal{C}}

\newcommand{\cL}{\mathcal{L}}

\newcommand{\tr}{\mathrm{tr}}

\newcommand{\cA}{\mathcal{A}}

\newcommand\lh{\cL(\cH)}

\newcommand{\sh}{\cS\ler{\cH}}

\author[Gergely Bunth]{Gergely Bunth}
\address{Gergely Bunth, HUN-REN Alfr\'ed R\'enyi Institute of Mathematics\\ Re\'altanoda u. 13-15.\\Budapest H-1053\\ Hungary\\ and Department of Analysis and Operations Research, Institute of Mathematics, Budapest University of Technology and Economics\\ M\H{u}egyetem rkp. 3. \\ Budapest H-1111 \\ Hungary}
\email{bunth.gergely@renyi.hu}

\author[J\'ozsef Pitrik]{J\'ozsef Pitrik}
\address{J\'ozsef Pitrik, HUN-REN Wigner Research Centre for Physics\\ Budapest H-1525, Hungary\\ and HUN-REN Alfr\'ed R\'enyi Institute of Mathematics\\ Re\'altanoda u. 13-15.\\ Budapest H-1053\\ Hungary\\ and Department of Analysis and Operations Research, Institute of Mathematics \\Budapest University of Technology and Economics\\ M\H{u}egyetem rkp. 3. \\ Budapest H-1111\\ Hungary}
\email{pitrik.jozsef@renyi.hu}

\author[Tam\'as Titkos]{Tam\'as Titkos}
\address{Tam\'as Titkos, Corvinus University of Budapest\\ Department of Mathematics\\ Fővám tér 13-15.\\ Budapest H-1093\\Hungary\\ 
and \\ HUN-REN Alfr\'ed R\'enyi Institute of Mathematics\\ Re\'altanoda u. 13-15.\\ Budapest H-1053\\Hungary}
\email{tamas.titkos@uni-corvinus.hu}

\author[D\'aniel Virosztek]{D\'aniel Virosztek}
\address{D\'aniel Virosztek, HUN-REN Alfr\'ed R\'enyi Institute of Mathematics\\ Re\'altanoda u. 13-15.\\Budapest H-1053\\ Hungary}
\email{virosztek.daniel@renyi.hu}

\date{}

\subjclass[2020]{Primary: 49Q22; 81P16. Secondary: 81Q10.}

\keywords{quantum optimal transport, quantum channels, Petz recovery map, swap transposition}

\thanks{Bunth was supported by the Momentum Program of the Hungarian Academy of Sciences (grant no. LP2021-15/2021); Pitrik was supported by the “Frontline” Research Excellence Programme of the Hungarian National Research, Development and Innovation Office - NKFIH (grant no. KKP133827) and by the Momentum Program of the Hungarian Academy of Sciences (grant no. LP2021-15/2021);T. Titkos is supported by the Hungarian National Research, Development and Innovation Office (NKFIH) under grant agreements no. K134944 and no. Excellence\_151232, and by the Momentum program of the Hungarian Academy of Sciences under grant agreement no. LP2021-15/2021. D. Virosztek is supported by the Momentum program of the Hungarian Academy of Sciences under grant agreement no. LP2021-15/2021, by the Hungarian National Research, Development and Innovation Office (NKFIH) under grant agreement no. Excellence\_151232, and partially supported by the ERC Synergy Grant No. 810115.}

\title[The swap transpose and the Petz recovery map]{The swap transpose on couplings translates to Petz' recovery map on quantum channels}

\begin{document}

\begin{abstract}
In \cite{DPT-AHP}, De Palma and Trevisan described a one-to-one correspondence between quantum couplings and quantum channels realizing transport between states. The aim of this short note is to demonstrate that taking the Petz recovery map for a given channel and initial state is precisely the counterpart of the swap transpose operation on couplings. That is, the swap transpose of the coupling $\Pi_{\Phi}$ corresponding to the channel $\Phi$ and initial state $\rho$ is the coupling $\Pi_{rec}$ corresponding to the Petz recovery map $\Phi_{rec}.$
\end{abstract}

\maketitle

\section{Introduction}

\subsection{Motivation and main result}

The theory of optimal transportation has been an active field of research in recent decades, and it became one of the central topics in analysis with intimate connections to partial differential equations, fluid mechanics, probability theory, stochastic analysis and differential equations, and Riemannian geometry \cite{Ben-Bren00,bgl,Brenier-polar-en,Butkovsky,Hairer2,Hairer3,JKO-97a,JKO-97b,JKO-98,Otto-dissipative-PME,Sturm6,Sturm5}. Applied sciences such as machine learning, bioinformatics, and image processing also benefit from techniques derived from or inspired by optimal transport theory, see e.g. \cite{PeyreCuturi,bia1,SK1,Peyre1,Peyre2,MIalk7,arjovsky-wgan,CuturiLightspeed,Klein2025,Bunne2024,dss}.
\par
Moreover, a large variety of non-commutative or quantum versions of optimal transport theory have been proposed in the last few decades. Starting with the spectral distance of Connes and Lott \cite{Connes-Lott}, several essentially different concepts have been introduced. We mention the free probability approach of Biane and Voiculescu \cite{BianeVoiculescu}, see also \cite{Shlyakhtenko-free-monotone-transport,Shlyakhtenko-free-transport,Shlyakhtenko-free-Wass}, the semiclassical approach \cite{ZyczkowskiSlomczynski1,ZyczkowskiSlomczynski2}, the dynamical approach of Carlen, Maas, Datta, Rouzé, and Wirth  \cite{CarlenMaas-2,CarlenMaas-3,CarlenMaas-4, DattaRouze1, DattaRouze2, wirth-NC-transport-metric, Wirth-dual}, and many concepts based on quantum couplings, see the works of Caglioti, Golse, Mouhot, and Paul \cite{CagliotiGolsePaul, CagliotiGolsePaul-towardsqot, GolseMouhotPaul, GolsePaul-Schrodinger,GolsePaul-OTapproach, GolseTPaul-pseudometrics,GolsePaul-wavepackets}, Friedland, Eckstein, Cole, and \.Zyczkowski \cite{Friedland-Eckstein-Cole-Zyczkowski-MK,Cole-Eckstein-Friedland-Zyczkowski-QOT,BistronEcksteinZyczkowski}, Duvenhage \cite{Duvenhage1, Duvenhage-ext-quantum-det-bal,Duvenhage-quad-Wass-vNA, Duvenhage3}, and Beatty, Fran\'ca, Pitrik, and T\'oth \cite{TothPitrik,toth-pitrik-2025, beatty-franca}. A substantial part of the above mentioned current approaches to non-commutative optimal transport is covered by the book \cite{OTQS-book} and the surveys \cite{beatty2025wasserstein, Trevisan-review-QOT}.
\par
This note is concerned with a concept where quantum channels realize the transport, and quantum couplings are in a one-to-one correspondence with channels. This concept was introduced by De Palma and Trevisan \cite{DPT-AHP, DPT-lecture-notes}, and it is a quantum counterpart of the classical transport theory with quadratic cost --- see also \cite{BPTV-metric-24, BPTV-p-Wass, BPTV-Kantorovich, wirth-2025-triangle} for further recent developments in this direction. In this setting, we prove in an elementary way that two natural operations, the swap transpose on couplings, and taking the Petz recovery map \cite{Petz1986, Petz88, Accardi-Cecchini-82} on channels, are two sides of the same coin. We phrase the precise statement in Theorem \ref{thm:Petz-swap}--- all the necessary notions and notations will be introduced in Subsection \ref{subsec:basic-notions-notation}.

It is important to note that finite-dimensional versions of our results, with different conventions and terminology, have already been established. In \cite[Theorem 5.2]{Duvenhage2025quantumdetailed}, Duvenhage, Oerder, and van den Heuvel prove that the Accardi-Cecchini dual \cite{Accardi-Cecchini-82} of a completely positive map (on operators acting on a finite dimensional Hilbert space) coincides with the dual corresponding to swapping (without the transpose, in contrast to the convention we use). Furthermore, in \cite{Parzygnat-James-bayes-2023}, Parzygnat and Fullwood proved essentially the same finite-dimensional result in the context of Bayes' rules in the quantum setting --- see in particular eq. (14) and eq. (25) there. 

\paragraph*{{ \bf Acknowledgement.}} We thank Rocco Duvenhage and Arthur Parzygnat for drawing our attention to their closely related works \cite{Duvenhage2025quantumdetailed} and \cite{Parzygnat-James-bayes-2023} after the appearence of the first version of our manuscript on arXiv, and for their detailed explanations of the connections. We are also gratefut to Melchior Wirth for his comments on a preliminary version of this manuscript, and in particular for introducing us to related results in the von Neumann algebraic setting.

\subsection{Basic notions, notation} \label{subsec:basic-notions-notation}

Let $\cH$ be a separable complex Hilbert space, and let $\lh^{sa}$ denote the set of self-adjoint but not necessarily bounded linear operators on $\cH.$ Let $\cS(\cH)$ stand for the set of states, that is, positive trace-class operators on $\cH$ with unit trace. The space of all bounded operators on $\cH$ is denoted by $\cB(\cH),$ and we recall that the collection of trace-class operators on $\cH$ is denoted by $\mathcal{T}_1(\cH)$ and defined by $\mathcal{T}_1(\cH)= \lers{X \in \cB(\cH) \, \middle| \, \tr_{\cH}[\sqrt{X^*X}] < \infty}.$ Similarly, $\mathcal{T}_2(\cH)$ stands for the set of Hilbert-Schmidt operators defined by $\mathcal{T}_2(\cH)= \lers{X \in \cB(\cH) \, \middle| \, \tr_{\cH}[X^*X] < \infty}.$ The transpose $A^T$ of a bounded linear operator $A \in \cB(\cH)$ is the bounded operator on the dual space $\cH^*$ defined by the identity $\varphi(Ax)=\ler{A^T \varphi}(x)$ for all $x \in \cH$ and $\varphi \in \cH^*.$ The support of $A \in \cB(\cH)$ is the closure of its range, that is, $\mathrm{supp}(A)=\mathrm{cl}\ler{\mathrm{ran}(A)}.$ Let $\cK$ and $\cH$ be separable Hilbert spaces, and let $\Phi: \mathcal{T}_1(\cK) \to \mathcal{T}_1(\cH)$ be a bounded linear map. Then, its adjoint $\Phi^{\dagger}: \cB(\cH) \to \cB(\cK)$ is defined by the requirement that
\begin{align} \label{eq:dagger-adjoint-def}
\tr_{\cH}\lesq{\Phi(\rho)A}=\tr_{\cK}\lesq{\rho \, \Phi^{\dagger}(A)} \text{ for all } \rho \in \mathcal{T}_1(\cK) \text{ and } A \in \cB(\cH).
\end{align}
A completely positive and trace preserving linear map $\Phi: \mathcal{T}_1(\cK) \to \mathcal{T}_1(\cH)$ is called a quantum channel.
\par
The correspondence between quantum channels and quantum couplings described in \cite{DPT-AHP} is the following: for a given initial state $\rho \in \sh$ and a channel $\Phi: \mathcal{T}_1\ler{\mathrm{supp}(\rho)} \to \mathcal{T}_1(\cH)$ defined on trace-class operators acting on the support of $\rho,$ the quantum coupling $\Pi_{\Phi}$ of $\rho$ and $\Phi(\rho)$ is
\begin{align} \label{eq:Pi-phi-def}
    \Pi_{\Phi}=\ler{\Phi \otimes \mathrm{id}_{\mathcal{T}_1\ler{\cH^*}}} \ler{\Ket{\sqrt{\rho}}\Bra{\sqrt{\rho}}},
\end{align}
where $\Ket{\sqrt{\rho}}\Bra{\sqrt{\rho}}$ is the canonical purification (see \cite{Holevo}) of the state $\rho \in \sh.$ Here and throughout this note, we will use the canonical linear isomorphism between $\mathcal{T}_2(\cH)$ and $\hohc,$ which is the linear extension of the map
\begin{align}
    \psi \otimes \eta \mapsto \ket{\psi} \circ \eta \qquad \ler{\psi \in \cH, \eta \in \cH^*}. \nonumber
\end{align}
Accordingly, for an $X \in \mathcal{T}_2(\cH),$ the symbol $\Ket{X}$ denotes the map $\C \ni z \mapsto z X \in \mathcal{T}_2(\cH) \simeq \hohc,$ while $\Bra{X}$ stands for the map $\mathcal{T}_2(\cH) \in Y \mapsto \tr_{\cH}\lesq{X^*Y},$ where $X^*$ is the adjoint of $X.$ One can check that $\Pi_{\Phi}$ given by \eqref{eq:Pi-phi-def} is a positive operator on $\hohc$ such that its marginals are $\Phi(\rho)$ and $\rho^T,$ respectively, that is, $\tr_{\cH^*}\lesq{\Pi_{\Phi}}=\Phi(\rho)$ and $\tr_{\cH}\lesq{\Pi_{\Phi}}=\rho^T.$ Accordingly, the set of all quantum couplings of the states $\rho, \omega \in \sh$ (denoted by $\cC(\rho, \omega)$) was defined in \cite{DPT-AHP} by
\begin{align} \label{eq:q-coup-def}
\cC\ler{\rho, \omega}=\lers{\Pi \in \cS\ler{\cH \otimes \cH^*} \, \middle| \, \tr_{\cH^*} \lesq{\Pi}=\omega, \,  \tr_{\cH} \lesq{\Pi}=\rho^T}.
\end{align}
We recall that the \emph{energy} $E_A(\rho)$ of a quantum state $\rho$ admitting the spectral resolution $\rho=\sum_{j=0}^{\infty} \lambda_j \ket{\varphi_j}\bra{\varphi_j}$ with respect to the observable $A \in \cL(\cH)^{sa}$ is given by
\begin{align} 
    E_A(\rho)=\sum_{j=0}^{\infty} \lambda_j \norm{A \varphi_j}^2 \in [0,+\infty] \nonumber
\end{align}
if $\varphi_j$ is in the domain of $A$ for every $j,$ and $E_A(\rho)=+\infty$ otherwise.
Given a finite set of observable quantities represented by the self-adjoint (and possibly unbounded) operators $A_1, \dots, A_K \in \cL(\cH)^{sa},$ the transport cost $C(\Pi)$ of the coupling $\Pi \in \cC(\rho, \omega)$ is the sum of the energies of $\Pi$ with respect to the observables $A_k \otimes I^T - I \otimes A_k^T.$ That is,
\begin{align}
    C(\Pi)=\sum_{k=1}^K E_{A_k \otimes I^T - I \otimes A_k^T}(\Pi) \in [0, +\infty]. \nonumber
\end{align}
To avoid technicalities that would obscure the main idea, we assume from now on that all the observable quantities are bounded, that is, $A_1,\dots, A_K \in \cB(\cH)^{sa}.$ In this case, the cost of a coupling takes the simpler form
\begin{align}
    C(\Pi)=\tr_{\hohc}\lesq{\Pi\ler{\sum_{k=1}^K \ler{A_k \otimes I^T - I \otimes A_k^T}^2}}, \nonumber
\end{align}
and hence we introduce the cost operator $C_{\cA}$ corresponding to $\cA=\{A_1,\dots, A_K\}$ the following way:
\begin{align} \label{eq:C-A-def}
    C_{\cA}:=\sum_{k=1}^K \ler{A_k \otimes I^T - I \otimes A_k^T}^2.
\end{align}
As the computation in \cite[Sec. 5.2]{DPT-lecture-notes} shows, in the case of bounded observables, the cost of the coupling $\Pi_\Phi$ corresponding to the channel $\Phi$ (see \eqref{eq:Pi-phi-def}) can be expressed referring to the channel instead of the coupling:
\begin{align} \label{eq:cost-in-terms-of-channel}
    C\ler{\Pi_\Phi}&=\sum_{k=1}^K \tr_{\hohc}\lesq{\ler{\Phi \otimes \mathrm{id}_{\mathcal{T}_1\ler{\cH^*}}} \ler{\Ket{\sqrt{\rho}}\Bra{\sqrt{\rho}}} \ler{A_k \otimes I^T - I \otimes A_k^T}^2} \nonumber \\
    &=\sum_{k=1}^K \tr_{\hohc}\lesq{\Ket{\sqrt{\rho}}\Bra{\sqrt{\rho}}\ler{\Phi^{\dagger}\ler{A_k^2} \otimes I^T + \Phi^{\dagger}(I)\otimes \ler{A_k^T}^2-2 \Phi^{\dagger}\ler{A_k} \otimes A_k^T }} \nonumber \\
    &=\sum_{k=1}^K \ler{\Bra{\sqrt{\rho}} \Phi^{\dagger}(A_k^2) \otimes I^T \Ket{\sqrt{\rho}} + \Bra{\sqrt{\rho}} I \otimes \ler{A_k^T}^2 \Ket{\sqrt{\rho}} -2
    \Bra{\sqrt{\rho}} \Phi^{\dagger}(A_k) \otimes A_k^T  \Ket{\sqrt{\rho}}} \nonumber \\
    &=\sum_{k=1}^K \ler{\tr_{\cH}\lesq{\rho \, \Phi^{\dagger}(A_k^2)}
    +\tr_{\cH}\lesq{\rho A_k^2} -2 \tr_{\cH}\lesq{\sqrt{\rho}A_k\sqrt{\rho}\Phi^{\dagger}(A_k)}} \nonumber \\
    &=\sum_{k=1}^K \ler{\tr_{\cH}\lesq{\Phi(\rho) A_k^2}
    +\tr_{\cH}\lesq{\rho A_k^2} -2 \tr_{\cH}\lesq{\sqrt{\rho}A_k\sqrt{\rho}\Phi^{\dagger}(A_k)}}.
\end{align}
In the above computation, the identity $\ler{A \otimes B^T} \Ket{X}=\Ket{AXB},$ which is valid for all $A,B \in \cB(\cH)$ and $X \in \mathcal{T}_2(\cH),$ see \cite[Lemma 1]{DPT-AHP}, has been used several times, as well as the cyclicity of the trace for the product of two Hilbert-Schmidt operators.
\par
The \emph{quadratic quantum Wasserstein distance} $D_{\cA}(\rho, \omega)$ of the states $\rho, \omega \in \sh$ corresponding to the observables $\cA=\{A_1, \dots, A_K\}$ is defined by
\be \label{eq:qw-dist-def}
D_{\cA}^2(\rho, \omega)=\inf\lers{C(\Pi)\, \middle| \, \Pi \in \cC\ler{\rho, \omega}},
\ee
and if $A_1, \dots, A_k \in \cB(\cH)^{sa},$ then 
\begin{align}
    D_{\cA}^2(\rho, \omega)=\inf\lers{\tr_{\hohc}\lesq{\Pi C_{\cA}}\, \middle| \, \Pi \in \cC\ler{\rho, \omega}}. \nonumber
\end{align}
Let us introduce the notation 
\begin{align} \label{eq:CPTP-notation-def}
    \text{CPTP}(\rho, \omega)=\lers{\Phi: \mathcal{T}_1(\mathrm{supp}(\rho)) \to \mathcal{T}_1(\cH), \, \Phi \text{ is completely positive and trace preserving, and } \Phi(\rho)=\omega}.
\end{align}
By \eqref{eq:cost-in-terms-of-channel}, one can rewrite \eqref{eq:qw-dist-def} as follows:

\begin{align} \label{eq:qw-dist-channels}
    D_{\cA}^2(\rho, \omega)=\inf\lers{
    \sum_{k=1}^K \ler{\tr_{\cH}\lesq{\omega A_k^2}
    +\tr_{\cH}\lesq{\rho A_k^2} -2 \tr_{\cH}\lesq{\sqrt{\rho}A_k\sqrt{\rho}\Phi^{\dagger}(A_k)}}
    \, \middle| \, \Phi \in \text{CPTP}(\rho, \omega)}.
\end{align}
We will denote by $\langle C_{\cA}\rangle_{\Phi}^{\rho,\omega}$ the cost of the quantum channel $\Phi$ sending $\rho$ to $\omega$ with respect to the observables $\cA=\{A_1, \dots,A_K\},$ that is,
\begin{align} \label{eq:C-Phi-rho-omega-def}
    \langle C_{\cA}\rangle_{\Phi}^{\rho,\omega}
    =\sum_{k=1}^K \ler{\tr_{\cH}\lesq{\omega A_k^2}
    +\tr_{\cH}\lesq{\rho A_k^2} -2 \tr_{\cH}\lesq{\sqrt{\rho}A_k\sqrt{\rho}\Phi^{\dagger}(A_k)}}.
\end{align}
Recall that the Wasserstein distance is symmetric \cite[Remark 8]{DPT-AHP}, and the easiest way to demonstrate this is to use  the swap transpose \cite[Proposition 4]{DPT-AHP}. Let $(\cdot)^{ST}$ denote the swap transposition acting either on $\mathcal{T}_1(\hohc)$ or on $\cB(\hohc)$ which is the linear extension of the map $X \otimes Y^T \mapsto Y \otimes X^T.$ The adjoint (in the sense of \eqref{eq:dagger-adjoint-def}) of the operation $\Pi \mapsto \Pi^{ST}$ on $\mathcal{T}_1(\hohc)$ is the same operation on $\cB(\hohc)$ as seen from the computation
\begin{align}
    \tr_{\hohc}\lesq{\ler{\omega \otimes \rho^T}^{ST}\ler{A \otimes B^T}}&=\tr_{\hohc}\lesq{\ler{\rho \otimes \omega^T}\ler{A \otimes B^T}} \nonumber\\ 
    =\tr_{\cH}\lesq{\rho A} \tr_{\cH}\lesq{\omega B}
    &=\tr_{\hohc}\lesq{\ler{\omega \otimes \rho^T}\ler{B \otimes A^T}}
    =\tr_{\hohc}\lesq{\ler{\omega \otimes \rho^T}\ler{A \otimes B^T}^{ST}}, \nonumber
\end{align}
where $\rho,\omega \in \mathcal{T}_1(\cH),$ and $A,B \in \cB(\cH),$ and we used the identity $\tr_{\cH^*}\lesq{X^T Y^T}=\tr_{\cH}\lesq{XY}.$
One can see from the spectral decomposition of the quadratic cost operator $C_{\cA}$ defined in \eqref{eq:C-A-def} that it is invariant under $(\cdot)^{ST}$. However, for any $\Pi \in \cC\ler{\rho, \omega}$ it is immediate that $\Pi^{ST} \in \cC\ler{\omega, \rho}$. Note that although transposition is not completely positive it is still positive. We conclude that for any sates $\rho$ and $\omega$ and any coupling $\Pi \in \cC\ler{\rho, \omega}$ and any $\cA=\lers{A_1,\dots, A_K},$ the following holds:
\begin{align}
C\ler{\Pi^{ST}}
=\tr_{\hohc} \lesq{\Pi^{ST} \, C_{\cA}}
=\tr_{\hohc} \lesq{\Pi \, C_{\cA}^{ST}}
=\tr_{\hohc} \lesq{\Pi \, C_{\cA}}
=C(\Pi). \nonumber
\end{align}
That is, $\Pi$ and $\Pi^{ST}$ are of the same cost.
\par
Let $\rho \in \sh$ be a state, and let $\Phi: \mathcal{T}_1(\mathrm{supp}(\rho)) \to \mathcal{T}_1(\cH)$ be a quantum channel, that is, a completely positive and trace-preserving map. Then the Petz recovery map (see, e.g., \cite[Section 12.3]{Wilde-QIT}) for $\rho$ and $\Phi$ is the completely positive and trace-preserving map $\Phi_{\text{rec}}^{\rho}$ defined by
\begin{align} \label{eq:Petz-rec-def}
\Phi_{\text{rec}}^{\rho}: \mathcal{T}_1(\mathrm{supp}(\Phi(\rho))) \to \mathcal{T}_1(\cH); \, X \mapsto
\Phi_{\text{rec}}^{\rho}(X):=\rho^{1/2}\Phi^\dagger\ler{\Phi(\rho)^{-1/2}X\Phi(\rho)^{-1/2}}\rho^{1/2},
\end{align}
where we denoted simply by $X$ the operator $X \oplus 0_{\ler{\mathrm{supp}(\Phi(\rho))}^{\perp}}$ on the right-hand side of \eqref{eq:Petz-rec-def}, with a slight abuse of notation. In the above formula \eqref{eq:Petz-rec-def}, the operator $\Phi(\rho)^{-1/2}$ is the pseudo-inverse of $\Phi(\rho)^{1/2},$ that is, $\Phi(\rho)^{-1/2}=\sum_{j: \mu_j>0} \mu_j^{-1/2}\ket{g_j}\bra{g_j}$ where $\Phi(\rho)=\sum_{j=1}^{\infty} \mu_j\ket{g_j}\bra{g_j}$ is the spectral decomposition of $\Phi(\rho),$ and the eigenvalues are ordered decreasingly: $\mu_1 \geq \mu_2 \geq \mu_3 \geq \dots.$ 
It is important to note that $\Phi(\rho)^{-1/2}$ is bounded if and only if $\Phi(\rho)$ is of finite rank. In this case, $\Phi(\rho)^{-1/2}X\Phi(\rho)^{-1/2}$ is also bounded for any $X \in \mathcal{T}_1(\mathrm{supp}(\Phi(\rho))),$ and hence the definition \eqref{eq:Petz-rec-def} makes perfect sense. However, if $\Phi(\rho)$ is of infinite rank, then $\Phi(\rho)^{-1/2}X\Phi(\rho)^{-1/2}$ can be unbounded for some $X \in \mathcal{T}_1(\mathrm{supp}(\Phi(\rho))),$ which means that the definition \eqref{eq:Petz-rec-def} does not make sense without further discussion as $\Phi^{\dagger}$ is a map from $\cB(\cH)$ to $\cB(\mathrm{supp}(\rho)).$ So assume that $\mathrm{rank}(\Phi(\rho))=\infty,$ and let us define the set
\begin{align} 
    S_{\Phi(\rho)} := \lers{X \in \mathcal{T}_1(\mathrm{supp}(\Phi(\rho))) \, \middle| \, X=\sum_{i,j=1}^N x_{i,j}\ket{g_i}\bra{g_j} \text{ for some }N\in \N \text{ and } x_{1,1}, x_{1,2} \dots, x_{N,N} \in \C}. \nonumber
\end{align}
Note that $S_{\Phi(\rho)}$ is a dense linear subspace of $\mathcal{T}_1(\mathrm{supp}(\Phi(\rho)))$ in the trace norm topology. Indeed, finite rank operators are dense in $\mathcal{T}_1(\mathrm{supp}(\Phi(\rho))),$ and one can approximate rank-$1$ projections there by elements of $S_{\Phi(\rho)}$ arbitrarily well. 
The next step is to show that $\Phi(\rho)^{-1/2}X\Phi(\rho)^{-1/2}$ is a bounded operator on $\cH$ whenever $X \in S_{\Phi(\rho)}.$ By the definition of the pseudo-inverse, $\Phi(\rho)^{-1/2}X\Phi(\rho)^{-1/2}$ annihilates every vector in the ortho-complement of $\mathrm{supp}(\Phi(\rho)),$ so we can restrict our attention to $\mathrm{supp}(\Phi(\rho)).$ The space of all finite linear combinations of the eigenvectors of $\Phi(\rho)$ corresponding to positive eigenvalues, that is, $\mathrm{linspan}\ler{\{g_j\}_{j: \mu_j>0}}$ is a dense subspace of $\mathrm{supp}(\Phi(\rho)),$ and it is easy to check that
\begin{align} \label{eq:check-on-linspan}
    \Phi(\rho)^{-1/2}X\Phi(\rho)^{-1/2}=\sum_{i,j=1}^N \mu_i^{-1/2}\mu_j^{-1/2} x_{i,j} \ket{g_i}\bra{g_j}
\end{align}
there when $X=\sum_{i,j=1}^N x_{i,j}\ket{g_i}\bra{g_j}.$ The right-hand side of \eqref{eq:check-on-linspan} is of finite rank and hence bounded. Therefore, $\Phi(\rho)^{-1/2}X\Phi(\rho)^{-1/2}$ has a unique bounded linear extension on the whole $\mathrm{supp}(\Phi(\rho)).$ 
So $\Phi_{\text{rec}}^{\rho}$ is well-defined on $S_{\Phi(\rho)},$ moreover, it is completely positive and trace preserving on there: the complete positivity follows from the fact that is it a composition of completely positive maps, while the trace preserving property follows from the direct computation
\begin{align}
    \tr_{\cH}\lesq{\Phi_{\text{rec}}^{\rho}(X)}
    &=\tr_{\cH}\lesq{\rho^{1/2}\Phi^\dagger\ler{\Phi(\rho)^{-1/2}X\Phi(\rho)^{-1/2}}\rho^{1/2}}
    =\tr_{\cH}\lesq{\rho \,\Phi^{\dagger}\ler{\Phi(\rho)^{-1/2}X\Phi(\rho)^{-1/2}}} \nonumber \\
    &=\tr_{\cH} \lesq{\Phi(\rho) \Phi(\rho)^{-1/2}X\Phi(\rho)^{-1/2}}=\tr_{\cH}[X], \nonumber
\end{align}
where we used the cyclic property of the trace for the product of two Hilbert-Schmidt operators, and the definition \eqref{eq:dagger-adjoint-def} of $\Phi^{\dagger}.$ It is known that positive and trace preserving maps are contractive with respect to the trace norm, see, e.g., \cite[Theorem 9.2]{nielsen-chuang}. In particular, $\Phi_{\text{rec}}^{\rho}$ is bounded with respect to the trace norm on the dense linear subspace $S_{\Phi(\rho)}\subset \mathcal{T}_1(\mathrm{supp}(\Phi(\rho))),$ and hence it has a unique continuous extension on the whole $\mathcal{T}_1(\mathrm{supp}(\Phi(\rho))).$ So the map $\Phi_{\text{rec}}^{\rho}: \mathcal{T}_1(\mathrm{supp}(\Phi(\rho))) \to \mathcal{T}_1(\cH)$ is this unique continuous extension if $\Phi(\rho)$ is of infinite rank. This extension inherits the trace-preserving property and the complete positivity of $\Phi_{\text{rec}}^{\rho}$ on $S_{\Phi(\rho)}.$ 
\par
The Petz recovery map \eqref{eq:Petz-rec-def}
reverses the action of $\Phi$ on $\rho,$ that is, $\Phi_{rec}^{\rho}\ler{\Phi(\rho)}=\rho,$ and it is of the same cost as $\Phi$ as a transport plan between $\rho$ and $\Phi(\rho).$
More precisely, if $\rho$ and $\omega$ are states on $\cH$ and $\Phi$ is a quantum channel sending $\rho$ to $\omega,$  then, for any finite collection of observables $\cA=\{A_1,\dots, A_K\},$ the cost of the Petz recovery map $\Phi_{\text{rec}}$ is the same as the cost of $\Phi,$ that is,
\begin{align} \label{eq:costs-are-equal}
\langle C_{\cA} \rangle_{\Phi_{rec}}^{\omega,\rho}=\langle C_{\cA}\rangle_\Phi^{\rho,\omega},
\end{align}
where the costs appearing in \eqref{eq:costs-are-equal} are defined in \eqref{eq:C-Phi-rho-omega-def}.
Indeed, by \eqref{eq:Petz-rec-def} one gets
\begin{align}
\Phi_{rec}(\omega^{1/2} A_k\omega^{1/2})=\rho^{1/2}\Phi^{\dagger}\ler{A_k}\rho^{1/2} \nonumber
\end{align}
for all $k=1,\dots,K.$ 
Consequently, 
\begin{align}
\tr_{\cH}\lesq{\rho^{1/2} A_k \rho^{1/2}\Phi^\dagger\ler{A_k}}
=\tr_{\cH}\lesq{A_k \Phi_{rec}\ler{\omega^{1/2}A_k\omega^{1/2}}}
=\tr_{\cH}\lesq{\Phi_{rec}^{\dagger}( A_k) \omega^{1/2}A_k\omega^{1/2}} \nonumber
\end{align}
for all $k=1,\dots, K$, and therefore
\begin{align} \label{eq:covariances-are-equal}
\sum_{k=1}^K
\tr_{\cH}\lesq{\rho^{1/2} A_k \rho^{1/2}\Phi^\dagger\ler{A_k}}=
\sum_{k=1}^K
\tr_{\cH}\lesq{\omega^{1/2}A_j\omega^{1/2}\Phi_{\text{rec}}^{\dagger}( A_j)}. 
\end{align}
Now the statement \eqref{eq:costs-are-equal} follows from \eqref{eq:covariances-are-equal} and \eqref{eq:C-Phi-rho-omega-def}.

The next result shows that not only the cost of the Petz recovery map is the same as that of the original channel, but the Petz recovery channel translates to swap transposition when considering quantum couplings to describe transport plans instead of quantum channels. The precise statement, using the convention that $\omega^{-1/2}$ denotes the square root of the pseudo-inverse of $\omega,$ reads as follows. 

\begin{theorem} \label{thm:Petz-swap}
Let $\Phi$ be a quantum channel acting on $\sh$ which sends $\rho$ to $\omega,$ and let $\Pi_\Phi \in \cC\ler{\rho,\omega}$ be the corresponding coupling. Let $\Phi_{rec}$ denote the Petz recovery map given by
\begin{align} 
    \Phi_{rec}(X)=\rho^{1/2}\Phi^{\dagger}\ler{\omega^{-1/2} X \omega^{-1/2}}\rho^{1/2}, \nonumber
\end{align}
and let $\Pi_{rec}$ denote the coupling that corresponds to $\Phi_{rec}$ according to \eqref{eq:Pi-phi-def}. Then
\begin{align} 
\ler{\Pi_\Phi}^{ST}=\Pi_{rec}. \nonumber
\end{align}
\end{theorem}

\begin{proof}
Assume that $\rho$ and $\omega$ admit the following spectral resolutions: $\rho=\sum_{i=1}^{\infty} \lambda_i \ket{e_i}\bra{e_i}$ and $\omega=\sum_{j=1}^{\infty} \mu_j \ket{g_j}\bra{g_j}.$ Let $\{f_i\}$ and $\{h_j\}$ denote the dual bases of $\{e_i\}$ and $\{g_j\},$ respectively, in $\cH^*.$ Then 
\begin{align} 
    \Pi_{\Phi}
    =\ler{\Phi \otimes \mathrm{id}_{\mathcal{T}_1(\cH^*)}}\ler{\Ket{\sqrt{\rho}}\Bra{\sqrt{\rho}}}
    =\ler{\Phi \otimes \mathrm{id}_{\mathcal{T}_1(\cH^*)}}\ler{\Ket{\sum_i \sqrt{\lambda_i} e_i \otimes f_i}\Bra{\sum_j \sqrt{\lambda_j} e_j \otimes f_j}} \nonumber \\
    =\ler{\Phi \otimes \mathrm{id}_{\mathcal{T}_1(\cH^*)}}\ler{\sum_{i,j}\sqrt{\lambda_i \lambda_j} \ket{e_i}\bra{e_j} \otimes \ler{\ket{e_j}\bra{e_i}}^T}
    =\sum_{i,j} \sqrt{\lambda_i \lambda_j}\Phi\ler{\ket{e_i}\bra{e_j}}\otimes \ler{\ket{e_j}\bra{e_i}}^T, \nonumber
\end{align}
where by the infinite sum $\sum_{i,j}\sqrt{\lambda_i \lambda_j} \ket{e_i}\bra{e_j} \otimes \ler{\ket{e_j}\bra{e_i}}^T$ we mean the trace norm limit 
$$
\lim_{M,N \to \infty}^{\norm{\cdot}_1} \sum_{i=1}^M \sum_{j=1}^N \sqrt{\lambda_i \lambda_j} \ket{e_i}\bra{e_j} \otimes \ler{\ket{e_j}\bra{e_i}}^T,
$$ and similarly, 
$$
\sum_{i,j} \sqrt{\lambda_i \lambda_j}\Phi\ler{\ket{e_i}\bra{e_j}}\otimes \ler{\ket{e_j}\bra{e_i}}^T= \lim_{M,N \to \infty}^{\norm{\cdot}_1} \sum_{i=1}^M \sum_{j=1}^N  \sqrt{\lambda_i \lambda_j}\Phi\ler{\ket{e_i}\bra{e_j}}\otimes \ler{\ket{e_j}\bra{e_i}}^T.
$$
Therefore,
\begin{align} \label{eq:Pi-Phi-final}
    \ler{\Pi_\Phi}^{ST}
    =\sum_{i,j} \sqrt{\lambda_i \lambda_j}\ket{e_j}\bra{e_i} \otimes \ler{\Phi\ler{\ket{e_i}\bra{e_j}}}^T
    =\sum_{i,j} \rho^{1/2}\ket{e_j}\bra{e_i}\rho^{1/2} \otimes \ler{\Phi\ler{\ket{e_i}\bra{e_j}}}^T.
\end{align}
On the other hand, 
\begin{align} 
    \Pi_{rec}
    =\ler{\Phi_{rec} \otimes \mathrm{id}_{\mathcal{T}_1(\cH^*)}}\ler{\Ket{\sqrt{\omega}}\Bra{\sqrt{\omega}}}
    =\ler{\Phi_{rec} \otimes \mathrm{id}_{\mathcal{T}_1(\cH^*)}}\ler{\Ket{\sum_k \sqrt{\mu_k} g_k \otimes h_k}\Bra{\sum_l \sqrt{\mu_l} g_l \otimes h_l}} \nonumber \\
    =\ler{\Phi_{rec} \otimes \mathrm{id}_{\mathcal{T}_1(\cH^*)}}\ler{\sum_{k,l}\sqrt{\mu_k \mu_l} \ket{g_k}\bra{g_l} \otimes \ler{\ket{g_l}\bra{g_k}}^T}
    =\sum_{k,l} \sqrt{\mu_k \mu_l}\Phi_{rec}\ler{\ket{g_k}\bra{g_l}}\otimes \ler{\ket{g_l}\bra{g_k}}^T \nonumber \\
    =\sum_{k,l} \sqrt{\mu_k \mu_l}\rho^{1/2}\Phi^{\dagger}\ler{\omega^{-1/2}\ket{g_k}\bra{g_l}\omega^{-1/2}}\rho^{1/2}\otimes \ler{\ket{g_l}\bra{g_k}}^T  \nonumber \\
    =\sum_{k,l} \sqrt{\mu_k \mu_l}\rho^{1/2}\Phi^{\dagger}\ler{\mu_k^{-1/2}\mu_l^{-1/2}\ket{g_k}\bra{g_l}}\rho^{1/2}\otimes \ler{\ket{g_l}\bra{g_k}}^T \nonumber \\
    =\sum_{k,l} \rho^{1/2}\Phi^{\dagger}\ler{\ket{g_k}\bra{g_l}}\rho^{1/2}\otimes \ler{\ket{g_l}\bra{g_k}}^T. \nonumber
\end{align}
We shall utilize now that 
\begin{align}
    \ket{g_k}\bra{g_l}=\sum_{r,s} \bra{e_r}\ket{g_k}\bra{g_l}\ket{e_s}\ket{e_r}\bra{e_s} \text{ and }
    \ket{g_l}\bra{g_k}=\sum_{m,n} \bra{e_m}\ket{g_l}\bra{g_k}\ket{e_n}\ket{e_m}\bra{e_n}, \nonumber
\end{align}
where, again, the infinite sums stand for the trace-norm limit of the finite partial sums.
Accordingly,
\begin{align} \label{eq:Pi-rec-final}
    \sum_{k,l} \rho^{1/2}\Phi^{\dagger}\ler{\ket{g_k}\bra{g_l}}\rho^{1/2}\otimes \ler{\ket{g_l}\bra{g_k}}^T \nonumber  \\
    =\sum_{k,l} \sum_{r,s} \sum_{m,n} \bra{e_r}\ket{g_k}\bra{g_l}\ket{e_s} \bra{e_m}\ket{g_l}\bra{g_k}\ket{e_n} \rho^{1/2}\Phi^{\dagger}\ler{\ket{e_r}\bra{e_s}}\rho^{1/2} \otimes \ler{\ket{e_m}\bra{e_n}}^T \nonumber \\
    =\sum_{r,s} \sum_{m,n} \sum_{k,l} \bra{e_r}\ket{g_k}\bra{g_k}\ket{e_n} \bra{e_m}\ket{g_l}\bra{g_l}\ket{e_s} \rho^{1/2}\Phi^{\dagger}\ler{\ket{e_r}\bra{e_s}}\rho^{1/2} \otimes \ler{\ket{e_m}\bra{e_n}}^T \nonumber \\
    =\sum_{r,s} \sum_{m,n} \delta_{r,n}\delta_{m,s} \rho^{1/2}\Phi^{\dagger}\ler{\ket{e_r}\bra{e_s}}\rho^{1/2} \otimes \ler{\ket{e_m}\bra{e_n}}^T
    =\sum_{r,s} \rho^{1/2}\Phi^{\dagger}\ler{\ket{e_r}\bra{e_s}}\rho^{1/2} \otimes \ler{\ket{e_s}\bra{e_r}}^T. 
\end{align}
Let us compare the formula for $\ler{\Pi_\Phi}^{ST}$ given by the right-hand side of \eqref{eq:Pi-Phi-final} and the formula for $\Pi_{rec}$ given by the right-hand side of \eqref{eq:Pi-rec-final}. To do so, consider the matrix unit 
\begin{align} 
    E_{m,n,r,s}:=\ket{e_m}\bra{e_n} \otimes \ler{\ket{e_s}\bra{e_r}}^T \nonumber
\end{align}
for arbitrary indices $m,n,r,$ and $s,$ and let us compute 
\begin{align}
    \tr_{\hohc} \lesq{\ler{\Pi_\Phi}^{ST} E_{m,n,r,s}} \text{ and } \tr_{\hohc} \lesq{\Pi_{rec} \, E_{m,n,r,s}}. \nonumber
\end{align}
By \eqref{eq:Pi-Phi-final}, we get
\begin{align} \label{eq:Pi-Phi-test}
    \tr_{\hohc} \lesq{\ler{\Pi_\Phi}^{ST} E_{m,n,r,s}}
    =\tr_{\hohc} \lesq{\ler{\sum_{i,j} \rho^{1/2}\ket{e_j}\bra{e_i}\rho^{1/2} \otimes \ler{\Phi\ler{\ket{e_i}\bra{e_j}}}^T}\ler{\ket{e_m}\bra{e_n} \otimes \ler{\ket{e_s}\bra{e_r}}^T}} \nonumber \\
    =\sum_{i,j} \bra{e_n} \rho^{1/2}\ket{e_j}\bra{e_i}\rho^{1/2} \ket{e_m}\cdot \bra{e_r}\Phi\ler{\ket{e_i}\bra{e_j}}\ket{e_s}
    =\sum_{i,j} \lambda_j^{1/2}\delta_{j,n}\lambda_i^{1/2} \delta_{i,m} \cdot \bra{e_r}\Phi\ler{\ket{e_i}\bra{e_j}}\ket{e_s} \nonumber \\
    =\ler{\lambda_m \lambda_n}^{1/2} \bra{e_r}\Phi\ler{\ket{e_m}\bra{e_n}}\ket{e_s}.
\end{align}
On the other hand, by \eqref{eq:Pi-rec-final}, we get
\begin{align} \label{eq:Pi-rec-test}
    \tr_{\hohc} \lesq{\Pi_{rec} \, E_{m,n,r,s}}
    =\tr_{\hohc} \lesq{\ler{\sum_{i,j} \rho^{1/2}\Phi^{\dagger}\ler{\ket{e_i}\bra{e_j}}\rho^{1/2} \otimes \ler{\ket{e_j}\bra{e_i}}^T} \ler{\ket{e_m}\bra{e_n} \otimes \ler{\ket{e_s}\bra{e_r}}^T}} \nonumber \\
    =\sum_{i,j} \bra{e_n}\rho^{1/2}\Phi^{\dagger}\ler{\ket{e_i}\bra{e_j}}\rho^{1/2} \ket{e_m}\cdot \bra{e_r}\ket{e_j}\bra{e_i}\ket{e_s}
    =\sum_{i,j} \ler{\lambda_m \lambda_n}^{1/2} \bra{e_n}\Phi^{\dagger}\ler{\ket{e_i}\bra{e_j}} \ket{e_m}\cdot \delta_{r,j}\delta_{i,s} \nonumber \\
    =\ler{\lambda_m \lambda_n}^{1/2} \bra{e_n}\Phi^{\dagger}\ler{\ket{e_s}\bra{e_r}} \ket{e_m}
    =\ler{\lambda_m \lambda_n}^{1/2} \tr_{\cH}\lesq{\ket{e_m} \bra{e_n} \Phi^{\dagger}\ler{\ket{e_s}\bra{e_r}}} \nonumber \\
    =\ler{\lambda_m \lambda_n}^{1/2} \tr_{\cH}\lesq{\Phi\ler{\ket{e_m} \bra{e_n}} \ket{e_s}\bra{e_r}} 
    =\ler{\lambda_m \lambda_n}^{1/2} \bra{e_r}\Phi\ler{\ket{e_m}\bra{e_n}}\ket{e_s}.
\end{align}
The final formulae of \eqref{eq:Pi-Phi-test} and \eqref{eq:Pi-rec-test} coincide, hence $\Pi_{rec}=\ler{\Pi_{\Phi}}^{ST}$ indeed.
\end{proof}

\begin{small}
\bibliographystyle{plainurl}  
\bibliography{references.bib}
\end{small}

\end{document}